\documentclass[conference]{IEEEtran}
\IEEEoverridecommandlockouts
\usepackage{cite}
\usepackage{amsmath,amssymb,amsfonts}
\usepackage{amsthm}
\usepackage{algorithmic}
\usepackage{graphicx}
\usepackage{textcomp}
\usepackage{xcolor}
\usepackage{tikz, color, soul}
\usepackage{multirow}
\newtheorem{theorem}{Theorem}
\newtheorem{definition}{Definition}

\newtheorem{corollary}{Corollary}
\usetikzlibrary{shapes.geometric,arrows,positioning}
\tikzstyle{decision} = [diamond, draw, fill=blue!20, 
    text width=5em, text badly centered, node distance=4cm, inner sep=0pt]
\tikzstyle{block} = [rectangle, draw, fill=blue!20,  text centered, rounded corners, minimum height=4em]
\tikzstyle{line} = [draw, -latex']
\tikzstyle{cloud} = [draw, ellipse,fill=red!20, node distance=6.6cm,
    minimum height=1em]
\tikzstyle{algorithm} = [rectangle, draw, fill=green!20,  text centered, rounded corners, minimum height=4em, minimum width =6em]
\tikzstyle{initialization} = [rectangle, draw,   text centered, minimum height=4em, minimum width =6em]
\def\BibTeX{{\rm B\kern-.05em{\sc i\kern-.025em b}\kern-.08em
    T\kern-.1667em\lower.7ex\hbox{E}\kern-.125emX}}
    
    \tikzstyle{block}=[draw, rectangle, minimum height=1cm, text width=1.5cm, text centered, draw=darkgray, font=\small]
\tikzstyle{block_medium}=[draw, rectangle, minimum height=1.5cm, text width=2cm, text centered, draw=darkgray, font=\small]
\tikzstyle{block_large}=[draw, rectangle, minimum height=2cm, text width=2cm, text centered, draw=darkgray, font=\small]
\tikzstyle{line} = [draw, -latex]

\newtheorem{Proposition}{Proposition}

\newtheorem{Example}{Example}

\def\F{{\mathcal F}}

\def\M{{\mathcal M}}

\def\X{{\mathcal X}}
\def\Y{{\mathcal Y}}

\begin{document}

\title{Non-adaptive and two-stage coding over the Z-channel
%\thanks{This work was supported by ....}
}

\author{\IEEEauthorblockN{Alexey Lebedev}
\IEEEauthorblockA{\textit{Kharkevich Institute}\\
\textit{for Information Transmission Problems} \\
\textit{Russian Academy of Sciences}\\
Moscow, Russia \\
al\_lebed95@mail.ru}\\
\IEEEauthorblockN{Vladimir Lebedev}
\IEEEauthorblockA{\textit{Kharkevich Institute}\\
\textit{for Information Transmission Problems} \\
\textit{Russian Academy of Science}\\
Moscow, Russia \\
lebedev37@mail.ru}
\and
\IEEEauthorblockN{Ilya Vorobyev}
\IEEEauthorblockA{\textit{Center for Computational and Data-Intensive}\\
\textit{Science and Engineering} \\
\textit{Skolkovo Institute of Science and Technology}\\
Moscow, Russia \\
vorobyev.i.v@yandex.ru}\\
\IEEEauthorblockN{Christian Deppe}
\IEEEauthorblockA{\textit{Institute for Communications Engineering} \\
\textit{Department of Electrical and Computer Engineering}\\
\textit{Technical University of Munich }\\
Munich, Germany \\
christian.deppe@tum.de}
}

\maketitle

\begin{abstract}
In this paper, we developed new coding strategies for the Z-channel. In particular, we look at the case with two-stage encoding. In this case, the encoder uses noiseless feedback once and adjusts the further encoding strategy based on the previous partial output of the channel. Nevertheless, the developed codes improve the known results with full feedback for small length and 1 error. A tool for the two-stage strategy is the development of a new optimality condition for non-adaptive codes. 
\end{abstract}


\section{Introduction}
Symmetric error-correcting codes have been studied repeatedly. Here one assumes that every symbol can be disturbed and can be decoded falsely as another symbol. However, it also makes sense to consider asymmetric error-correcting codes for combinatorial models. For this purpose in this paper we consider the Z-channel.  
In the Z-channel, a transmitted 1 can be received as a 0, but a transmitted 0 is always correctly received as a 0. This channel is a pure combinatorial model. One considers asymmetric codes with and without feedback. Unlike the probabilistic channel models (see \cite{S48}), combinatorial codes are improved by feedback (see \cite{B68}). Although symmetric channel models have been more thoroughly studied, there has been research on asymmetric codes for a long time (see survey \cite{K81}). 

In this paper, we create coding strategies for the Z-channel with noiseless feedback. This task is equivalent to a variant of the so-called Ulam game \cite{ulam1991adventures}, the half-lie game.
The half-lie game was first defined in \cite{rivest1980coping}. It is a two-person game. Here, a player tries to find an element $x \in \mathcal{M}$ by asking $n$ yes-no questions. The questions are of the form: Is $x \in A$ for any $A \subseteq \mathcal{M}$? The opponent knows $x$ and may lie at most $t$ times if the correct answer to the question is yes. In the original Ulam game \cite{ulam1991adventures}, the opponent is also allowed to lie if the correct answer is no. The original Ulam game was first described by Renyi \cite{renyi61} and was later used to compute the error-correcting capacity for symmetric binary channels by Berlekamp \cite{B68}. In his autobiography, Ulam asked in \cite{ulam1991adventures} for special solutions for $t=1,2,3$ for $M=1000000$. The optimal strategies were determined for $t=1$ in \cite{P87}, for $t=2$ in \cite{G90} and for $t=3$ in \cite{D00} for each $M$. In addition, a table of optimal strategies was created in \cite{D02}.
In the book \cite{Cicalese13}, one can find a good overview of the results. For fixed chosen $t$, the cardinality of the maximal set $\mathcal{M}$ is asymptotically $2^{n+t}t! n^{-t}$ to find a successful strategy for the half-lie game.
This was first shown for $t=1$ in~\cite{cicalese2000optimal} and then for any $t$ in~\cite{dumitriu2004halfliar,spencer2003halflie}.
%\textcolor{red}{add results about Ulam game for small t}

The games described above are adaptive games, that is, each question is answered immediately and the next question depends on the previously received answers. These games correspond to a coding strategy with immediate feedback after each symbol. 
In \cite{DS05} the authors have shown that for a constant
number of errors, it is sufficient to use the feedback only
once to asymptotically transmit the same number of messages for the $Z$-channel.
In this paper we consider such coding strategies with one feedback. More precisely, we assume that the sender receives feedback only after a sequence of length $n_1$ bits on all of the bits received so far and then sends additional $n_2$ bits.

In Section~\ref{definitions} we give the basic definitions and survey the known results. Then in Section~\ref{sec::bounds for small length} we show how to use the linear programming bound to find asymmetric error-correcting codes with a maximum number of free points and apply this for small $n$. Finally, in Section~\ref{two} we use the previously found codes with free points for a two-stage coding strategy. The obtained results are better than the previously known coding strategies with complete feedback for some length and one error. %For $n=6$ and $n=7$ we were able develop a better coding strategy with complete feedback.

\section{Basic definitions and known results}\label{definitions}

We consider communication over a binary Z-channel
with input alphabet $\X=\{0,1\}$ and output alphabet $\Y=\X$, where a word of length $n$ is sent
by the encoder. Here we assume that a 0 is always transmitted correctly. Furthermore, a 1 can either be transmitted correctly or changed to a 0 by the channel for at most $t$ times (see Fig.~\ref{fig1}).

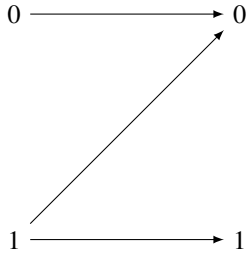
\begin{figure}[t]
\centering
	\begin{tikzpicture}
	\node (A) at (0,0) {0};
	\node (C) at (3,0) {0};
	\node (B) at (0,-3) {1};
	\node (D) at (3,-3) {1};
	\path[line] (A) -- (C);
	\path[line] (B) -- (D);
	\path[line] (B) -- (C);
	\end{tikzpicture}
	\caption{Z-channel}
	\label{fig1}
\end{figure}

The sender wants to transmit a message $m\in\M=\{1,2,\dots, M\}$ over the binary Z-channel with noiseless feedback. $\M$ is the message set and $\X$ is the coding alphabet. 
Suppose now, having sent $(x_1,\ldots ,x_{j-1})=x^{j-1}\in\X^{j-1}$, the encoder
knows the received letters $(y_1,\ldots ,y_{j-1})\in\Y^{j-1}$ before it sends the
next letter $x_j$ ($j=1,2,\ldots ,n$). In this case we are dealing with a channel with complete noiseless feedback.
A channel without feedback is a non-adaptive channel. 
We are considering here a case where the sender does not get a feedback after each bit. We assume that the sender first sends $n_1$ bits and only then receives feedback on the bits received so far. Depending on this, it sends $n_2$ more bits. Thereby $n_1+n_2=n$ is valid. 
This is called a channel with one feedback. 
An encoding function (algorithm) for a channel with one feedback is defined by
\[
c(m,y^{n_1})=((c_{1}(m),c_{2}(m,y^{n_1})),
\]
where $c_{1}:\M\to \X^{n_1}$, $c_{2}:\M\times\X^{n_1}\to \X^{n_2}$ and
\[
x^{n_1}=c_1(m);\;\;(x_{n_1+1}, \ldots x_{n}) = c_2(m, y^{n_1}).
\]
The receiver has a decoding function $d:\X^n\to \M$, which maps the received
sequence $y^n$
to a message.

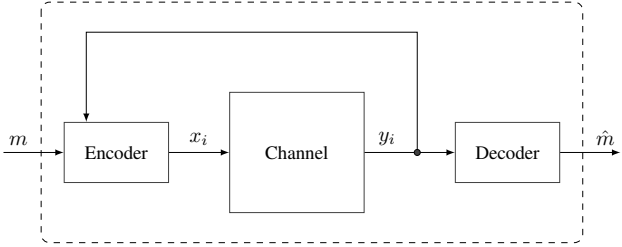
\begin{figure}[t]
\centering
\scalebox{0.8}{
  \begin{tikzpicture}
    \node[block_large] (c) at (0,0) {Channel};
    \node[block] (d) at (3.5,0) {Decoder};
    \node[block] (e1) at (-3,0) {Encoder};
    \node (b) at (-5,0) {};
    \node (u1) at (5.5,0) {};
    \path[line] (b) -- node[near start, above] {$m$} (e1);
    \path[line] (e1.east) -- node[above] {$x_i$} (c.west);
    \path[line] (c) -- node[near start, above] {$y_i$} (d);
    \path[line] (2,0) -- (2,2) -- (-3.5,2) -- ([xshift=-0.5cm]e1.north);
    \path[draw, dashed, rounded corners] (-4.25,2.5) -- (-4.25,-1.5) -- (4.75,-1.5) -- (4.75,2.5) -- cycle;
    \path[line] (d.east) -- node[near end, above] {$\hat{m}$} (u1);
    \node[draw, circle, minimum size=1mm, inner sep=0pt, outer sep=0pt, fill=darkgray] at (2,0) {};
\end{tikzpicture}
}
  \caption{Channel with feedback}
  \label{fig:channel_feedback}
\end{figure}
We say that an $(n,M,t)_{pf}=(c(m, y^{n_1}), d(y^{n}))$ two-stage coding algorithm is successful if a decoding algorithm $d$ correctly obtains a message $m$ encoded with a function $c$ 
%there exists an encoding function of the sender and a decoding function of the receiver, such 
%that the receiver decodes the message of the sender correctly
if fewer than $t$
errors happen. We emphasize that the receiver should be able to correctly decode the message for all possible error configurations with at most $t$ errors, i.e. it is a combinatorial model, not a probabilistic one.
In our paper, all of the described coding algorithms we have considered are successful.

%\section{Basic definitions and known results}\label{definitions}
For our analysis, we need to recall the results about the non-adaptive Z-channel.
For any two vectors $a^n,b^n\in\{0,1\}^n$ we introduce the notation $N(a^n,b^n):=|\{i: a_i=0, b_i=1\}|$. We define the Z-distance $d_Z$ as $d_Z(a^n,b^n):=\max\{N(a^n,b^n),N(b^n,a^n)\}$.

\iffalse
\begin{definition}
Let $a^n,b^n\in\{0,1\}^n$ then we set
\begin{enumerate}
    \item $N(a^n,b^n):=|\{i: a_i=0, b_i=1\}|$,
    \item $d_Z(a^n,b^n):=\max\{N(a^n,b^n),N(b^n,a^n)\}$,
    % \item $d_H(x^n,y^n):=N(x^n,y^n)+N(y^n,x^n)$,
    % \item $w_H(x^n):=|\{ i: x_i=1\}|$
\end{enumerate}
We call $d_Z$ the Z-distance.
\end{definition}
\fi
The well-known Hamming distance $d_H$, which counts the number of different bits, can also be calculate with this notation: $d_H(a^n,b^n):=N(a^n,b^n)+N(b^n,a^n)$.
The (Hamming) weight of a sequence (codeword) is the number of 1 in the sequence (codeword).
A non-adaptive asymmetric $(n,M)_Z$-code is
a subset $C\subset \{0,1\}^n$ with cardinality $|C|=M$. This code is called a (successful) asymmetric $(n,M,t)_Z$ code for the $Z$-channel if for all $a^n,b^n\in C$ with $a^n\neq b^n$ holds $d_Z(a^n,b^n)\geq 2t$. The decoder can correct up
to $t$ asymmetric errors with such a code. An $(n,M,t)_Z$-code is called optimal if there does not exist any $(n,M',t)_Z$-code with $M'>M$. By $M_Z(n,t)$ we denote the cardinality of an optimal $(n,M,t)_Z$-code. In coding theory, one would like to determine all values for optimal codes and find constructions for these codes. However, in many cases there are only upper and lower bounds. In papers \cite{al1997new,etzion1998greedy}, references for all these bounds are provided. 
% \textcolor{red}{we write about Table I twice. The first one is here, and the second is at the end of the section. It would probably be better to write everything about it in one place.}

\begin{table}
\begin{tabular}{l|l|l}
$n$ & lower bound for $M$ & upper bound for $M$\\\hline
1&1&1\\
2&2&2\\
3&2&2\\
4&4&4\\
5&6&6\\
6&12&12\\
7&18&18\\
8&36&36\\
9&62&62\\
10&112&117\\
11&198&210\\
12&379&410\\
\end{tabular}
\bigskip
\caption{Bounds for non-adaptive error-correcting codes}\label{tab1}
\end{table}
There can always be multiple constructions for optimal error-correcting codes. For our construction with one feedback we need a refinement of the optimality definition. While for symmetric error-correcting codes for each codeword the number of possible words at the receiver for given $t$ is always $M\choose t$, this is not true for asymmetric error-correcting codes. 
Let $a^n\in C$ be a codeword, then we denote by $B_t(a^n)$ the set of codewords $b^n$ such that $d_Z(a^n, b^n)\le t$. Let $C$ be an $(n,M,t)$ code then we denote by $\mathcal{F}:=\{0,1\}^n\setminus \bigcup_{a^n\in C} B_t(a^n)$ the set of free points and by $F=|\F|$ the number of free points.

We want to maximize the number of free points for a given cardinality of a code and denote an $(n, M, t)_Z$ code with $F$ free points as an $(n, M, F, t)_Z$ code. We call an $(n, M, F, t)_Z$ code $F$-optimal if there does not exist an $(n, M, F', t)_Z$-code such that $F'>F$. In the next section we construct $F$-optimal codes for small length $n$, $t=1$, $M\le M_Z(n, t)$.
In our paper we consider the case $t=1$. We summarize the best results for the non-adaptive case in Table~\ref{tab1}. Here, the results for the best coding strategies are from \cite{kim1959single,V73,delsarte1981bounds} and the best upper bounds are from \cite{weber1987new}. In addition, the adaptive case for $t=1$ was considered in \cite{cicalese2000optimal}. Some of our coding strategies in two stages are better than those developed there.
A solution for fixed errors when the block length approaches infinity was given for the adaptive case in \cite{dumitriu2004halfliar} and for the case with partial feedback (in two stages) in \cite{dumitriu2005two}.

\section{Asymmetric Error-Correcting Codes with Maximal Number of Free Points}\label{sec::bounds for small length}

In this section we use linear programming technique to obtain a bound on the number of free points in codes correcting one asymmetric error. We then apply this method to find these codes for small $n$.

\subsection{Linear Programming Upper Bound}
Delsarte and Piret\cite{delsarte1981bounds}, Klove\cite{klove1981upper}, and Weber et al. \cite{weber1987new} used linear programming to prove upper bounds on the size of non-adaptive codes correcting asymmetric errors. We modify this method to obtain upper bounds on the number of free points in a code of fixed size.

\begin{theorem}\label{th::linear programming upper bounds}
Let $n\ge 2t\ge 2$, $1\le M\le M_Z(n, t)$.
Define
$$
\overline{F}(n, M, t)=\max(2^n-\sum\limits_{i=0}^n z_i(i + 1))
$$
where the maximum is taken over the following constraints
\begin{enumerate}
    \item $z_i$ are nonnegative integers.
    \item $z_0=1$, $z_1=0$.
    \item 
    \begin{align*}
        \sum\limits_{i=1}^s\binom{n - w + i}{i}z_{w-i}+\sum\limits_{j=0}^{t-s}\binom{w+j}{j}z_{w+j}\le \binom{n}{w}
    \end{align*}
    for $0\le s\le t<w<n - t$.
    \item 
    \begin{multline*}
        \sum\limits_{j=s}^rz_jA_l(r-s, 2t+2, r-j)\\
        \le A_u(n+r-s, 2t+2, r)
    \end{multline*}
    for $0\le s\le r$.
    \item 
    \begin{multline*}
        \sum\limits_{j=s}^rz_{n-j}A_l(r-s, 2t+2, r-j)\\
        \le A_u(n+r-s, 2t+2, r)
    \end{multline*}

    for $0\le s \le r$.
    \item 
    \begin{multline*}
        \sum\limits_{i=1}^{s}\binom{n - w + i}{i}z_{w-i}+\sum\limits_{j=0}^{t-s}\binom{w+j}{j}z_{w+j}\\
        +\left(\binom{w+t-s+1}{w}-\binom{t+1}{t-s+1}\right.\\
        \cdot\left.\left\lfloor\frac{w+t-s+1}{t+1}\right\rfloor\right)
         z_{w+t-s+1}\le \binom{n}{w}
    \end{multline*}
    for $0\le s\le t<w<n - t$. 
    \begin{multline*}
        \sum\limits_{i=1}^{s}\binom{n - w + i}{i}z_{w-i}+\sum\limits_{j=0}^{t-s}\binom{w+j}{j}z_{w+j}\\
        +\left(\binom{n-w+s+1}{s+1}-\binom{t+1}{t-s}\right.\\
        \cdot\left.\left\lfloor\frac{n-w+s+1}{t+1}\right\rfloor\right)
        z_{w-s-1}\le \binom{n}{w}
    \end{multline*}
    for $0\le s\le t<w<n - t$. 
    \item
    $$
    \sum\limits_{i=0}^nz_i=M.
    $$
\end{enumerate}
Then the number of free points $F$ in an $(n, M, t)_Z$ code is at most $\overline{F}(n, M, t)$.
\end{theorem}

\begin{proof}
Let $z_i$, $0\le i \le n$ be the number of codewords of weight $i$ in a code of length $n$ correcting $t$ asymmetric errors. In papers\cite{delsarte1981bounds,klove1981upper,weber1987new} it was proved that $z_i$ should satisfy constraints 1), 3), 4), 5), and 6).
Obviously, a code with a maximal amount of free points should satisfy condition 2). The last condition fixes the size of the considered code. The maximized expression $\overline{F}(n, M, t)$ corresponds to the number of free points in a code with weight distribution $\{z_i\}$.
\end{proof}

\subsection{Computer search for $F$-optimal codes}\label{subsection::computer search}
We use an similar linear programming approach to find codes with a maximal number of free points. Fix length $n$, code size $M$, and number of corrected asymmetric errors $t$. Introduce $2^n$ binary variables $x_i$, corresponding to each possible codeword. For each point $p$ define a set $D_t(p)$ of codewords from which this point can be obtained with at most $t$ asymmetric errors. Then we introduce a constraint $\sum\limits_{i\in D_t(p)}x_i\le 1$. We are trying to maximize the number of free points $2^n-\sum\limits_{i=0}^n z_i(i + 1)$, where $z_i$ is a number of codewords of weight $i$. Note that the number of free points can be expressed through variables $x_i$. We add a constraint $\sum x_i=M$ to fix the size of a code. Note that any solution of this linear program (if it exists) gives an optimal number of free points for a given code length and size. To speed up the program, we add the constraints from Theorem~\ref{th::linear programming upper bounds}.

Still, this program has $2^n$ variables, and therefore, can be solved only for small $n$. In Table~\ref{table::optimal number of free points} we provide parameters of some $F$-optimal codes for $t=1$ and $n=6, 7, 8$ and $9$. Optimal weight distributions are depicted in Table~\ref{table::optimal weight distributions}.

The parameters of codes for $n=6, 8$ and $9$ coincide with the upper bounds given by Theorem~\ref{th::linear programming upper bounds}. %Weight distributions of the constructed codes coincide with solution of linear program from Theorem~\ref{th::linear programming upper bounds}.
For $n=7$ we have 48 points instead of 49 for the maximal size given by Theorem~\ref{th::linear programming upper bounds}, i.e. the upper bound from Theorem~\ref{th::linear programming upper bounds} is not achieved. All other values coincide with upper bounds. 

We note that the codes with optimal weight distributions for $n=7, 8$ were constructed in~\cite{delsarte1981bounds}. 
The code for $n=6$ is also known.  For $n=6, 8$ and all sizes $M\le M_Z(n, 1)$, an optimal code can be obtained from the code of maximal size with weight distribution described in Table~\ref{table::optimal weight distributions} by deleting $M_Z(n, d) - M$ codewords of maximal weight. For $n=7$ and $M=17$ we know two weight distributions with the same number of free points, which are achieved by two different $F$-optimal codes: $1+0+3+5+5+3+0+0$ and $1+0+3+5+6+1+1+0$. By deleting the codeword of maximal weight from the code with a second weight distribution, we can obtain $F$-optimal codes for all $M<17$. However, only the code with a first weight distribution can be extended to the code of size 18.

Our program works for all lengths $n<9$. The complexity is too great for longer block lengths. Since the $F$-optimal construction for $n=6, 8$ are given by nested codes, we restricted ourselves to search such families where $n=9$. This approach allows us to find a nested family such that a code with a maximal size of 62 has a weight distribution which is given in Table~\ref{table::optimal weight distributions}. The numbers of free points in this family coincides with upper bound given by Theorem~\ref{th::linear programming upper bounds} for all sizes $M$. It means that the constructed family is $F$-optimal. Note that the code of optimal size for $n=9$ constructed in~\cite{delsarte1981bounds} has 171 free points, while we have 177. 

\iffalse
\begin{table}[h]\label{table::upper bound for free points}
\begin{tabular}{|l|l|l|l|l|l|l|}
\hline
\multirow{2}{*}{$n=6$}  & size        & 12  & 11  & 10  & 9   & 8   \\ \cline{2-7} 
                        & free points & 16  & 23  & 28  & 33  & 38  \\ \hline
\multirow{2}{*}{$n=7$}  & size        & 18  & 17  & 16  & 15  & 14  \\ \cline{2-7} 
                        & free points & 48  & 56  & 62  & 68  & 73  \\ \hline
\multirow{2}{*}{$n=8$}  & size        & 36  & 35  & 34  & 33  & 32  \\ \cline{2-7} 
                        & free points & 76  & 85  & 92  & 99  & 106 \\ \hline
\end{tabular}
\caption{Optimal parameters of $(n, M, 1)$ codes}
\end{table}
\fi

\begin{table}[h]
\renewcommand{\arraystretch}{1.2}% for the vertical padding
\centering
\begin{tabular}{|c|c|c|c|c|c|c|}
\hline
\multirow{2}{*}{$n=6$}  & size $M$      & 12  & 11  & 10  & 9   & 8   \\ \cline{2-7} 
                        & free points $F$ & 16  & 23  & 28  & 33  & 38  \\ \hline
\multirow{2}{*}{$n=7$}  & size $M$       & 18  & 17  & 16  & 15  & 14  \\ \cline{2-7} 
                        & free points $F$ & 48  & 56  & 62  & 68  & 73  \\ \hline
\multirow{2}{*}{$n=8$}  & size $M$       & 36  & 35  & 34  & 33  & 32  \\ \cline{2-7} 
                        & free points $F$ & 76  & 85  & 92  & 99  & 106 \\ \hline
\multirow{2}{*}{$n=9$}  & size $M$       & 62  & 61  & 60  & 59  & 58  \\ \cline{2-7} 
                        & free points $F$ & 177 & 186 & 193 & 200 & 207 \\ \hline
%\multirow{2}{*}{$n=10$} & size        & 117 & 116 & 115 & 114 & 113 \\ \cline{2-7} 
                        %& free points & 327 & 338 & 347 & 356 & 365 \\ \hline
\end{tabular}
\bigskip
\caption{Optimal number of free points for $(n, M, 1)$ codes}\label{table::optimal number of free points}
\end{table}

\begin{table}[h]
\renewcommand{\arraystretch}{1.2}% for the vertical padding
\centering
\begin{tabular}{|c|c|}
\hline
length and size & \multicolumn{1}{c|}{weight distribution} \\ \hline
$n=6, M=12$      & 1+0+3+4+3+0+1                                             \\ \hline
$n=7, M=18$      & 1+0+3+5+5+3+1+0                                         \\ \hline
$n=7, M=17$      & 1+0+3+5+6+1+1+0                                           \\ \hline
$n=8, M=36$      & 1+0+4+8+10+8+4+0+1                                        \\ \hline
$n=9, M=62$      & 1+0+4+9+17+17+11+2+1+0                                    \\ \hline
%$n=10, M=117$     & 1+0+5+7+28+35+28+7+5+0+1                                  \\ \hline
\end{tabular}
\bigskip
\caption{Optimal weight distributions}
\label{table::optimal weight distributions}
\end{table}

\section{Codes with noiseless feedback in two stages}\label{two}

In this section we describe a two-stage strategy of transmission of messages over an Z-channel with 1 error. We split the codelength $n$ into two parts $n_1$ and $n_2$, $n=n_1+n_2$. At the first stage we transmit sequences of length $n_1$. There are two possibilities: either we have an error at the first stage or not. If there was an error, then we want to transmit the position of the error and some additional information. If there was no error then we use a non-adaptive code of length $n_2$, correcting one asymmetric error to transmit second part of the message.

Define a directed graph $G=(V, E)$ in the following way. Vertex set $V$ consists of $2^{n_1}$ vertexes. Each vertex is associated with a binary word of length $n_1$. We connect vertexes $v_1, v_2\in V$ with an arc $(v_1, v_2)$ if the codeword $v_2$ can be obtained from the codeword $v_1$ with exactly 1 asymmetric error, i.e. $N(v_1, v_2)=0$ and $N(v_1, v_2)=1$. Note that we don't connect the vertex with itself, i.e. there are no loops.

To each vertex $v\in V$ we appoint an $(n_2, M(v), F(v), 1)_Z$ code $C(v)$. For each free point $p$ of the code $C(v)$ we define a function $d_v(p): \mathcal{F}(C) \rightarrow U$, where $U=\{u\in V| (u, v)\in E\}$ In addition we demand that the number of points $p$ such that $d_v(p)=u$ is equal to $M(u)$ for each $u\in U$. It is possible if the following condition is satisfied for all $v\in V$
\begin{equation}\label{constraints}
    \sum\limits_{u:(u,v)\in E} M(u)\le F(v).
\end{equation}

Describe a decoding algorithm $d$. The first $n_1$ positions of the received codeword $y$ form a vector $v$, the last $n_2$ positions form a vector $v'$. If the vector $v'$ belongs to a ball $B_1(a)$ for some $a\in C(v)$, then we conclude that there was no error at the first $n_1$ symbols and decode the second part as $a$. Otherwise, the error has occurred at the first part. We decode the first part of the message as $u=d_{v}(v')$. Since there are $M(u)$ free points which can be decoded to $u$ from $v$, we transmit $M(u)$ messages in case of first part $u$. 

\begin{theorem}
Given a directed graph $G=(V, E)$ described above, let $C(v)$ be an $(n_2, M(v), F(v), 1)_Z$ code. If constraints~\eqref{constraints} are satisfied then 
\begin{equation}\label{number of transmitted words}
M=\sum\limits_{u\in V}M(u)
\end{equation}
messages can be transmitted.
\end{theorem}

For a specific choice of parameters $M(v)$ and $F(v)$, we have the following corollary.

\begin{corollary}
Let $M(v)=M_w$, $F(v)=F_w$ for all $v\in V$ such that the number of ones in $v$ is $w$, i.e. $M(v)$ and $F(v)$ depend only on the weight of $v$.  If constraints
\begin{equation}\label{simplified constaint}
    (n_1-w) M_{w+1}\le F_{w}
\end{equation}
are satisfied for all $w\in[0, n_1-1]$, then the number of transmitted messages is
\begin{equation}\label{number of transmitted words for special case}
   %\sum\limits_{w=0}^{n_1}\binom{n_1}{w}\min(M_w, \lfloor F_{w-1}/ (n_1-w+1)\rfloor), 
   M=\sum\limits_{w=0}^{n_1}\binom{n_1}{w}M_w.
\end{equation}
%where $F_{-1}=+\infty$.
\end{corollary}

\iffalse
Say that there are at most $p$ possibilities for the location of an error at the first stage. Then, if we use a code $C$ of length $n_2$ and  size $M_2$ with $F_2$ free points, then the number of transmitted messages equals
$$
M=\min(M_1\cdot M_2, M_1\cdot \lfloor F_2/p\rfloor). 
$$

The most obvious choice for the first stage is to take all possible messages of length $n_1$. In that case $M_1=2^{n_1}$, $p=n_1$. Sometimes it is better to throw out some words to reduce $p$. If we delete one codeword of weight 1, then there are always at most $p=n_1-1$ positions of an error. In that case we have $M_1=2^{n_1}-1$, $p=n_1-1$.

Using tables from Section~\ref{sec::bounds for small length} we find an optimal pair $(M_2, F_2)$ for a given $p$.
\fi

\begin{table}[h]
\renewcommand{\arraystretch}{1.2}% for the vertical padding
\begin{tabular}{|c|c|c|c|c|c|c|c|c|c|c|}
\hline
$n$  & 5 & 6  & 7  & 8  & 9  & 10  & 11  & 12 % & 13    
\\ \hline
$M$ (Cor 1) & 9 & 16 & 29 & 52 & 96 & 177 & 327 & 607 %& 1120 %& 2107 
\\ \hline
$M$ (Th 2) & 9 & 16 & 29 & 53 & 97 & $\ge 177$ & $\ge 329$ & $\ge 607$ %& \ge 1120 %& 2107 
\\ \hline
%Best & 9 & 16 & 28 & 52 & 96 & 177 & 327 & 607 & 1124 & 2108 \\ \hline
%Cicalese lower bound & ? & ? & ? & ? & ? & ? & ? & ? & ? & ? \\ \hline
\end{tabular}

\bigskip

\caption{Number of messages transmitted by proposed algorithm}\label{table::lower bounds for one feedback}
\end{table}

We provide number of transmitted messages given by Corollary 1 and Theorem~2 in Table~\ref{table::lower bounds for one feedback}. A dynamic programming technique was used to find an optimal assignment of $M_w$, $F_w$. In the following examples we give a detailed description of codes given by Corollary 1 and Theorem 2 for $n=9$ and $n=8$ correspondingly.

\begin{Example}
$n=9$, $M=96$.

Let  $n_1=5$, $n_2=4$. For vertexes with an associated binary word with weight $0$ and $1$ we appoint a $(4,2,12,1)_Z$ code $\{0000, 0011\}$. For weights $2$ and $3$ we appoint a $(4,3,9,1)_Z$ code $\{0000, 0011, 1100\}$. And for weights $4$ and $5$ - a $(4,4,4,1)_Z$ code $\{0000, 0011, 1100, 1111\}$ respectively.

Now we can check constraints $(n_1-w)\cdot M_{w+1}\leq F_w$ for $w\in[0, 4]$.
 
%For $v$ with weight $w$ the condition $(n_1-w)\cdot M_{w+1}\leq F_w$ holds.

$w=0$: $5\cdot2\leq12$,

$w=1$: $4\cdot3\leq12$

$w=2$: $3\cdot3\leq9$

$w=3$: $2\cdot4\leq9$

$w=4$: $1\cdot4\leq4$

Using formula~\eqref{number of transmitted words for special case} we compute that the total number of transmitted messages is $1\cdot2 +5\cdot2 +10\cdot3 +10\cdot3 +5\cdot4+1\cdot4=96$.
\end{Example}

\begin {Example}

$n=8$, $M=53$.

Let $n_1=6$, $n_2=2$. For the vertex $111111$ we appoint a $(2,2,0,1)_Z$ code $\{00, 11\}$. For the vertexes $111000$, $001110$, $010101$, $100011$, $100100$, $010010$, $001001$, $110000$, $010100$, $001000$, $000010$, $000001$ we appoint a $(2,0,4,1)_Z$ code. For the other vertexes we appoint a $(2,1,3,1)_Z$ code  $\{00\}$. One can verify that constraint~\eqref{constraints} is satisfied. One can check that the constraint for $v=101000$ is satisfied: There are 4 vertexes $\{111000, 101100, 101010, 101001\}$ such that $(u, v)\in E$.
\begin{align*}
    M(111000)&=0,\;M(101100)=1,\\
    M(101010)&=1,\;M(101001)=1.
\end{align*}
Therefore, their sum is not greater than $F(101000)=3$, i.e. the constraint for vertex $v=101000$ holds. The Constraints for the other vertexes $v$ can be verified in a similar manner.

The total number of transmitted messages is $2+0+51=53$.

\end{Example}

%\textcolor{red}{write example}

We compare our results with the strategy in \cite{cicalese2000optimal} with complete feedback. In this case for $M=2^m$ there is a strategy with $n=m-1+\lceil \log_2 (m+3)\rceil $ number of tests. Although here we use only two-stage coding, we get better values than in \cite{cicalese2000optimal} where full feedback is used. In particular, we improve the values for $n=5$ and $n\geq 8$.

\begin{table}[h]
\renewcommand{\arraystretch}{1.2}% for the vertical padding
\begin{tabular}{|c|c|c|c|c|c|c|c|c|c|c|}
\hline
$n$  & 5 & 6  & 7  & 8  & 9  & 10  & 11  & 12  & 13    \\ \hline
$M$ (\cite{cicalese2000optimal}) & 8 & 16 & 32 & 32 & 64 & 128 & 256 & 512 & 1024 %& 2048 
\\ \hline
%Best & 9 & 16 & 28 & 52 & 96 & 177 & 327 & 607 & 1124 & 2108 \\ \hline
%Cicalese lower bound & ? & ? & ? & ? & ? & ? & ? & ? & ? & ? \\ \hline
\end{tabular}
\bigskip
\caption{Lower bounds of codes with complete feedback}\label{table::adaptive strategies}
\end{table}

\section*{Conclusion}

In this work, we derived a new optimality condition for asymmetric error-correcting codes. We want to maximize the number of free points, i.e. such points that the distance from any codeword to a point is bigger than $t$.
%The used sequences are the sequences obtained from the codewords and possible errors. 
While for a symmetric channel this number remains the same for the same number of codewords, it can change for an asymmetric channel. We then use the codes with optimally many free points to develop two-stage coding strategies. Here we even improve the known results for the case with complete feedback for some lengths and one error.

\section*{Acknowledgment}
 Christian Deppe was supported by the Bundesministerium f\"ur Bildung und Forschung (BMBF) through Grant 16KIS1005. Alexey Lebedev was supported by RFBR and the National Science Foundation of Bulgaria (NSFB), project number 20-51-18002. Vladimir Lebedev's work was supported by RFBR and the National Science Foundation of Bulgaria (NSFB), project number 20-51-18002 and by RFBR and JSPS under Grant No. 20-51-50007. Ilya Vorobyev was supported by RFBR under Grant No. 20-01-00559, by RFBR and the National Science Foundation of Bulgaria (NSFB), project number 20-51-18002, by RFBR and JSPS under Grant No. 20-51-50007.

\bibliographystyle{IEEEtran}
\bibliography{mybib}

\end{document}